\documentclass{amsart}
\usepackage{amssymb, amsmath, amsthm}

\newcommand\R{\mathbb{R}}
\newcommand\C{\mathbb{C}}
\newcommand\N{\mathbb{N}}

\renewcommand\i{{\rm 1\kern -.3600em 1}}

\newcommand\vd{\delta}\newcommand\eps{\varepsilon}

\newtheorem{theorem}{Theorem}[section]

\newtheorem{lemma}[theorem]{Lemma}
\newtheorem{proposition}[theorem]{Proposition}

\newtheorem{remark}[theorem]{Remark}

\numberwithin{equation}{section}
\allowdisplaybreaks[3]

\begin{document}

\title[Kawasaki dynamics via generating functionals evolution]
    {Kawasaki dynamics in the continuum via generating functionals evolution}

\author{D. L. Finkelshtein}
\address{Institute of Mathematics, Ukrainian National Academy of
Sciences, 01601 Kiev, Ukraine} \curraddr{} \email{fdl@imath.kiev.ua}
\thanks{Financial support of DFG through SFB 701 (Bielefeld University),
German-Ukrainian Project 436 UKR 113/97 and FCT through
PTDC/MAT/100983/2008 and ISFL-1-209 are gratefully acknowledged.}

\author{Yu. G. Kondratiev}
\address{Fakult\"at f\"ur Mathematik, Universit\"at Bielefeld,
D 33615 Bielefeld, Germany; Forschungszentrum BiBoS, Universit\"at
Bielefeld, D 33615 Bielefeld, Germany}
\email{kondrat@mathematik.uni-bielefeld.de}

\author{M. J. Oliveira}
\address{Universidade Aberta, P 1269-001 Lisbon, Portugal;
CMAF, University of Lisbon, P 1649-003 Lisbon, Portugal}
\email{oliveira@cii.fc.ul.pt}

\date{DD/MM/2011}

\begin{abstract}
We construct the time evolution of Kawasaki dynamics for a spatial
infinite particle system in terms of generating functionals. This is
carried out by an Ovsjannikov-type result in a scale of Banach
spaces, which leads to a local (in time) solution. An application of
this approach to Vlasov-type scaling in terms of generating
functionals is considered as well.
\end{abstract}

\maketitle

\section{Introduction}

Originally, Bogoliubov generating functionals (GF for short) were
introduced by N.~N. Bogoliubov in \cite{Bog46} to define correlation
functions for statistical mechanics systems. Apart from this
specific application, and many others, GF are, by themselves, a
subject of interest in infinite dimensional analysis. This is
partially due to the fact that to a probability measure $\mu$
defined on the space $\Gamma$ of locally finite configurations
$\gamma\subset\R^d$ one may associate a GF
$$
B_\mu(\theta) :=\int_\Gamma d\mu(\gamma)\,\prod_{x\in \gamma }(1+\theta(x)),
$$
yielding an alternative method to study the stochastic dynamics of an
infinite particle system in the continuum by exploiting the close relation
between measures and GF \cite{FKO05,KoKuOl02}.

Existence and uniqueness results for the Kawasaki dynamics through GF arise
naturally from Picard-type approximations and a method suggested in
\cite[Appendix 2, A2.1]{GS58} in a scale of Banach spaces (see
e.g.~\cite[Theorem 2.5]{FKO11}). This method, originally presented for
equations with coefficients time independent, has been extended to an abstract
and general framework by T.~Yamanaka in \cite{Y60} and L.~V.~Ovsjannikov in
\cite{O65} in the linear case, and many applications were exposed by F.~Treves
in \cite{T68}. As an aside, within an analytical framework outside of our
setting, all these statements are very closely related to variants of the
abstract Cauchy-Kovalevskaya theorem. However, all these abstract forms only
yield a local solution, that is, a solution which is defined on a finite time
interval. Moreover, starting with an initial condition from a certain Banach
space, in general the solution evolves on larger Banach spaces.

As a particular application, this work concludes with the study of the
Vlasov-type scaling proposed in \cite{FKK10a} for general continuous particle
systems and accomplished in \cite{BKKK11} for the Kawasaki dynamics. The
general scheme proposed in \cite{FKK10a} for correlation functions yields a
limiting hierarchy which possesses a chaos preservation property, namely,
starting with a Poissonian (non-homogeneous) initial state this structural
property is preserved during the time evolution. In Section
\ref{Subsection3.2} the same problem is formulated in terms of GF and its
analysis is carried out by the general Ovsjannikov-type result in a scale of
Banach spaces presented in \cite[Theorem 4.3]{FKO11}.

\section{General Framework}\label{Section2}

In this section we briefly recall the concepts and results of combinatorial
harmonic analysis on configuration spaces and Bogoliubov generating
functionals needed throughout this work (for a detailed explanation see
\cite{KoKu99,KoKuOl02}).

\subsection{Harmonic analysis on configuration spaces}\label{Subsection2.1}

Let $\Gamma :=\Gamma _{\R^d}$ be the configuration space over $\mathbb{R}^d$,
$d\in\mathbb{N}$,
\[
\Gamma :=\left\{ \gamma \subset \mathbb{R}^d:\left| \gamma\cap\Lambda\right|
<\infty \hbox{
for every compact }\Lambda\subset \mathbb{R}^d\right\} ,
\]
where $\left| \cdot \right|$ denotes the cardinality of a set. We
identify each $\gamma \in \Gamma $ with the non-negative Radon measure
$\sum_{x\in \gamma }\delta_x$ on the Borel $\sigma$-algebra
$\mathcal{B}(\mathbb{R}^d)$, where $\delta_x$ is the Dirac measure with mass
at $x$, which allows to endow $\Gamma$ with the vague topology and the
corresponding Borel $\sigma$-algebra $\mathcal{B}(\Gamma)$.

For any $n\in\N_0:=\N\cup\{0\}$ let
\[
\Gamma^{(n)}:= \{ \gamma\in \Gamma: \vert \gamma\vert = n\},\ n\in \N,\quad \Gamma^{(0)} := \{\emptyset\}.
\]
Clearly, each $\Gamma^{(n)}$, $n\in\N$, can be identify with the
symmetrization of the set $ \{(x_1,...,x_n)\in (\R^d)^n: x_i\not=
x_j \hbox{ if } i\not= j\}$, which induces a natural (metrizable) topology on
$\Gamma^{(n)}$ and the corresponding Borel $\sigma$-algebra
$\mathcal{B}(\Gamma^{(n)})$. In particular, for the Lebesgue product
measure $(dx)^{\otimes n}$ fixed on $(\R^d)^n$, this identification yields a
measure $m^{(n)}$ on $(\Gamma^{(n)},\mathcal{B}(\Gamma^{(n)}))$. For $n=0$ we
set $m^{(0)}(\{\emptyset\}):=1$. This leads to the definition of the space of
finite configurations
\[
\Gamma_0 := \bigsqcup_{n=0}^\infty \Gamma^{(n)}
\]
endowed with the topology of disjoint union of topological spaces and the
corresponding Borel $\sigma$-algebra $\mathcal{B}(\Gamma_0)$, and to the
so-called Lebesgue-Poisson measure on $(\Gamma_0,\mathcal{B}(\Gamma_0))$,
\begin{equation}\label{LPm-def}
\lambda:=\lambda_{dx}:=\sum_{n=0}^\infty \frac{1}{n!} m^{(n)}.
\end{equation}

Let $\mathcal{B}_c(\R^d)$ be the set of all bounded Borel sets in
$\R^d$ and, for each $\Lambda\in \mathcal{B}_c(\R^d)$, let
$\Gamma_\Lambda := \{\eta\in \Gamma: \eta\subset \Lambda\}$.
Evidently $\Gamma_\Lambda = \bigsqcup_{n=0}^\infty
\Gamma_\Lambda^{(n)}$, where $\Gamma_\Lambda^{(n)}:= \Gamma_\Lambda
\cap \Gamma^{(n)}$, $n\in \N_0$. Given a complex-valued
$\mathcal{B}(\Gamma_0)$-measurable function $G$ such that
$G\!\!\upharpoonright _{\Gamma\backslash\Gamma_\Lambda}\equiv 0$ for
some $\Lambda \in \mathcal{B}_c(\R^d)$, the $K$-transform of $G$ is
a mapping $KG:\Gamma\to\mathbb{C}$ defined at each
$\gamma\in\Gamma$ by
\begin{equation}
(KG)(\gamma ):=\sum_{{\eta \subset \gamma}\atop{\vert\eta\vert < \infty} }
G(\eta ).
\label{Eq2.9}
\end{equation}
It has been shown in \cite{KoKu99} that the $K$-transform is a linear and
invertible mapping.

Let $\mathcal{M}_{\mathrm{fm}}^1(\Gamma)$ be the set of all
probability measures $\mu$ on $(\Gamma ,\mathcal{B}(\Gamma))$ with
finite local moments of all orders, i.e.,
\[
\int_\Gamma d\mu (\gamma)\, |\gamma\cap\Lambda |^n<\infty\quad
\mathrm{for\,\,all}\,\,n\in\N
\mathrm{\,\,and\,\,all\,\,} \Lambda \in \mathcal{B}_c(\R^d),
\]
and let $B_{\mathrm{bs}}(\Gamma_0)$ be the set of all complex-valued
bounded $\mathcal{B}(\Gamma_0)$-measurable functions with bounded
support, i.e., $G\!\!\upharpoonright _{\Gamma _0\backslash
\left(\bigsqcup_{n=0}^N\Gamma _\Lambda ^{(n)}\right) }\equiv 0$ for
some $N\in\N_0, \Lambda \in \mathcal{B}_c(\R^d)$. Given a
$\mu\in\mathcal{M}_{\mathrm{fm}}^1(\Gamma)$, the so-called
correlation measure $\rho_\mu$ corresponding to $\mu$ is a measure
on $(\Gamma _0,\mathcal{B}(\Gamma _0))$ defined for all $G\in
B_{\mathrm{bs}}(\Gamma_0)$ by
\begin{equation}
\int_{\Gamma _0}d\rho _\mu(\eta )\,G(\eta )=\int_\Gamma d\mu (\gamma)\, \left(
KG\right) (\gamma).  \label{Eq2.16}
\end{equation}
This definition implies, in particular, that $B_{\mathrm{bs}}(\Gamma_0)\subset
L^1(\Gamma_0,\rho_\mu)$.\footnote{Throughout this work all
$L^p$-spaces, $p\geq 1$, consist of complex-valued functions.} Moreover, still
by (\ref{Eq2.16}), on $B_{\mathrm{bs}}(\Gamma_0)$ the inequality $\Vert
KG\Vert_{L^1(\Gamma,\mu)}\leq \Vert G\Vert_{L^1(\Gamma_0,\rho_\mu)}$ holds, allowing
an extension of the $K$-transform to a bounded operator
$K:L^1(\Gamma_0,\rho_\mu)\to L^1(\Gamma,\mu)$ in such a way that
equality (\ref{Eq2.16}) still holds for any $G\in
L^1(\Gamma_0,\rho_\mu)$. For the extended operator the explicit form
(\ref{Eq2.9}) still holds, now $\mu$-a.e. In particular, for coherent states
$e_\lambda(f)$ of complex-valued $\mathcal{B}(\R^d)$-measurable functions $f$,
\begin{equation}\label{taquase6}
e_\lambda (f,\eta ):=\prod_{x\in \eta }f\left( x\right) ,\ \eta \in
\Gamma _0\!\setminus\!\{\emptyset\},\quad  e_\lambda (f,\emptyset ):=1.
\end{equation}
Additionally, if $f$ has compact support we have
\begin{equation}\label{Kcog}
\left( Ke_\lambda (f)\right) (\gamma )=\prod_{x\in \gamma }(1+f(x))
\end{equation}
for all $\gamma\in\Gamma$, while for functions $f$ such that
$e_\lambda(f)\in L^1(\Gamma_0,\rho_\mu)$ equality \eqref{Kcog} holds, but only
for $\mu$-a.a.~$\gamma\in\Gamma$. Concerning the Lebesgue-Poisson measure
\eqref{LPm-def}, we observe that $e_\lambda(f)\in L^p(\Gamma_0,\lambda)$
whenever $f\in L^p:=L^p(\R^d,dx)$ for some $p\geq 1$. In this case,
$\| e_\lambda(f)\|^p_{L^p}=\exp(\| f\|^p_{L^p})$. In particular, for $p=1$, in
addition we have
\begin{equation*}
\int_{\Gamma_0}d\lambda(\eta)\, e_\lambda(f,\eta)= \exp\left(\int_{\R^d}dx\,f(x)\right),
\end{equation*}
for all $f\in L^1$. For more details see \cite{KoKuOl00b}.

\subsection{Bogoliubov generating functionals}\label{Subsection2.2}

Given a probability measure $\mu$ on $(\Gamma, \mathcal{B} (\Gamma))$ the
so-called Bogoliubov generating functional (GF for short) $B_\mu$
corresponding to $\mu$ is the functional defined at each
$\mathcal{B}(\R^d)$-measurable function $\theta$ by
\begin{equation}
B_\mu(\theta) :=\int_\Gamma d\mu(\gamma)\,\prod_{x\in\gamma}(1+\theta (x)),
\label{Dima2}
\end{equation}
provided the right-hand side exists. It is clear from (\ref{Dima2})
that the domain of a GF $B_\mu$ depends on the underlying measure $\mu$ and,
conversely, the domain of $B_\mu$ reflects special properties over the
measure $\mu$. Throughout this work we will consider GF defined on
the whole complex $L^1$ space. This implies, in particular, that the underlying
measure $\mu$ has finite local exponential moments, i.e.,
\[
\int_\Gamma d\mu (\gamma )\, e^{\alpha|\gamma\cap\Lambda |}<\infty \quad
\hbox{for all}\,\,\alpha>0\,\,
\hbox{and all}\,\,\Lambda \in \mathcal{B}_c(\R^d),
\]
and thus $\mu\in\mathcal{M}_{\mathrm{fm}}^1(\Gamma)$.  According to the
previous subsection, this implies that to such a measure $\mu$ one may
associate the correlation measure $\rho_\mu$, which leads to a description of
the functional $B_\mu$ in terms of either the measure $\rho_\mu$:
\[
B_\mu(\theta)
= \int_\Gamma d\mu(\gamma)\,\left( Ke_\lambda (\theta)\right) (\gamma)
= \int_{\Gamma_0}d\rho_\mu(\eta)\, e_\lambda (\theta, \eta),
\]
or the so-called correlation function $k_\mu:=\frac{d\rho_\mu}{d\lambda}$
corresponding to the measure $\mu$, if $\rho_\mu$ is absolutely continuous
with respect to the Lebesgue--Poisson measure $\lambda$:
\begin{equation}\label{BF_via_cf}
B_\mu(\theta)=\int_{\Gamma_0}d\lambda(\eta)\, e_\lambda (\theta, \eta)k_\mu(\eta).
\end{equation}

Throughout this work we will assume, in addition, that GF are entire on the
$L^1$ space \cite{KoKuOl02}, which is a natural environment, namely, to
recover the notion of correlation function. For a generic entire functional
$B$ on $L^1$, this assumption implies that $B$ has a representation in terms
of its Taylor expansion,
$$
B(\theta_0+ z \theta)=\sum_{n=0}^\infty \frac{z^n}{n!}
d^nB(\theta_0;\theta ,...,\theta),\quad z\in\C, \theta\in L^1,
$$
being each differential $d^nB(\theta _0;\cdot), n\in\N, \theta _0\in L^1$
defined by a symmetric kernel $$\vd^n B(\theta_0;\cdot)\in
L^\infty (\R^{dn}):=L^\infty \bigl((\R^d)^n,(dx)^{\otimes n}\bigr),$$
called the variational derivative of $n$-th order of $B$ at the
point $\theta _0$. That is,
\begin{align}
d^nB(\theta _0;\theta _1,...,\theta _n) :&\!=\frac{\partial
^n}{\partial
z_1...\partial z_n}B\left( \theta _0+\sum_{i=1}^nz_i\theta _i\right) %
\Bigg\vert_{z_1=...=z_n=0} \label{taquase10}\\
&=\,:\int_{(\R^d)^n}dx_1\ldots dx_n\,\vd^n
B(\theta_0;x_1,\ldots,x_n)\prod_{i=1}^n\theta_i(x_i)\nonumber
\end{align}
for all $\theta _1,...,\theta _n\in L^1$. Moreover, the operator
norm of the bounded $n$-linear functional $d^nB(\theta_0;\cdot)$ is
equal to $\left\| \vd^nB(\theta_0;\cdot)\right\|_{L^\infty(\R^{dn})}$ and for
all $r>0$ one has
\begin{align}
\left\| \vd B (\theta_0;\cdot)\right\|_{L^\infty(\R^d)} &\leq
\frac{1}{r} \sup_{\|\theta^\prime \|_{L^1} \leq r}
|B(\theta_0+\theta^\prime)|\label{2011!}
\\
\intertext{and, for $n\geq 2$,} \left\|
\vd^nB(\theta_0;\cdot)\right\|_{L^\infty(\R^{dn})} &\leq n!
\left(\frac{e}{r}\right)^n \sup_{\|\theta^\prime \|_{L^1} \leq r}
|B(\theta_0+\theta^\prime)|. \label{Duarte}
\end{align}

In particular, if $B$ is an entire GF $B_\mu$ on $L^1$ then, in terms of the
underlying measure $\mu$, the entireness property of $B_\mu$ implies that the
correlation measure $\rho_\mu$ is absolutely continuous with respect to the
Lebesgue-Poisson measure $\lambda$ and the Radon-Nykodim derivative
$k_\mu=\dfrac{d\rho_\mu}{d\lambda}$ is given by
\[
k_\mu(\eta )= \vd^{\left|\eta\right|} B_\mu(0;\eta)
\quad \text{for $\lambda$-a.a.
}\eta \in \Gamma _0. 
\]

In what follows, for each $\alpha>0$, we consider the Banach space
$\mathcal{E}_\alpha$ of all entire functionals $B$ on $L^1$ such that
\[
\left\| B\right\| _\alpha :=\sup_{\theta \in L^1}
\left( \left|
B(\theta )\right| e^{-\frac{1}{\alpha} \left\| \theta \right\| _{L^1}}\right)
<\infty,
\]
see \cite{KoKuOl02}. This class of Banach spaces has the particularity that,
for each $\alpha_0>0$, the family
$\{\mathcal{E}_\alpha: 0<\alpha\leq\alpha_0\}$ is a scale of Banach spaces,
that is,
\[
\mathcal{E}_{\alpha''}\subseteq \mathcal{E}_{\alpha'},\quad
\|\cdot\|_{\alpha'}\leq \|\cdot\|_{\alpha''}
\]
for any pair $\alpha'$, $\alpha''$ such that $0<\alpha'< \alpha''\leq\alpha_0$.

\section{The Kawasaki dynamics}\label{Subsection3.1}

The Kawasaki dynamics is an example of a hopping particle model where, in this
case, particles randomly hop over the space $\R^d$ according to a rate
depending on the interaction between particles.  More precisely, let
$a:\R^d\to\left[0,+\infty\right)$ be an even and integrable function and let
$\phi:\R^d\to\left[0,+\infty\right]$ be a pair potential, that is, a
$\mathcal{B}(\R^d)$-measurable function such that $\phi(-x)=\phi(x)\in \R$ for
all $x\in \R^d\setminus\{0\}$, which we will assume to be integrable. A
particle located at a site $x$ in a given configuration $\gamma\in\Gamma$ hops
to a site $y$ according to a rate given by $a(x-y)\exp(-E(y,\gamma))$, where
$E(y,\gamma)$ is a relative energy of interaction between the site $y$ and the
configuration $\gamma$ defined by
\[
E(y,\gamma ):=\sum_{x\in \gamma }\phi (x-y)\in[0,+\infty].
\]
Informally, the behavior of such an infinite particle system is described by
\begin{equation}\label{Vl41}
(LF)(\gamma)= \sum_{x\in\gamma}\int_{\R^d}dy\,a(x-y)e^{-E(y,\gamma)}\left(F(\gamma\setminus\{x\}\cup\{y\})-F(\gamma)\right).
\end{equation}

Given an infinite particle system, as the Kawasaki dynamics, its time
evolution in terms of states is informally given by the so-called
Fokker-Planck equation,
\begin{equation}
\frac{d\mu_t}{dt}=L^*\mu_t, \qquad
\mu_t\!\bigm|_{t=0}=\mu_0\label{FokkerPlanck},
\end{equation}
where $L^*$ is the dual operator of $L$. Technically, the use of definition
\eqref{Eq2.16} allows an alternative approach to the study of
\eqref{FokkerPlanck} through the corresponding correlation functions
$k_t:=k_{\mu_t}$, $t\geq0$, provided they exist. This leads to the Cauchy
problem
\begin{equation*}
\frac \partial {\partial t}k_t=\hat L^*k_t,\quad
{k_t}_{|t=0}=k_0,
\end{equation*}
where $k_0$ is the correlation function corresponding to the initial
distribution $\mu_0$ and $\hat L^*$ is the dual operator of $\hat L:=K^{-1}LK$
in the sense
\begin{equation*}
\int_{\Gamma_0}d\lambda(\eta)\,(\hat LG)(\eta) k(\eta)=
\int_{\Gamma_0}d\lambda(\eta)\,G(\eta) (\hat L^*k)(\eta).
\end{equation*}
Through the representation \eqref{BF_via_cf}, this gives us a way to express
the dynamics also in terms of the GF $B_t$ corresponding to $\mu_t$, i.e.,
informally,
\begin{align}
\label{obtevol}\frac \partial {\partial t}B_t(\theta ) &=\int_{\Gamma _0}d\lambda(\eta)\,e_\lambda(\theta ,\eta )\left( \frac \partial{\partial t}k _t(\eta )\right)=\int_{\Gamma _0}d\lambda(\eta)\, e_\lambda(\theta ,\eta )(\hat L^*k_t)(\eta )\\
&=\int_{\Gamma _0}d\lambda(\eta)\,(\hat
Le_\lambda(\theta))(\eta)k_t(\eta)=:(\tilde L B_t)(\theta).\nonumber
\end{align}
This leads to the time evolution equation
\begin{equation}
\frac{\partial B_t}{\partial t}=\tilde LB_t,\label{evolgf}
\end{equation}
where, in the case of the Kawasaki dynamics, $\tilde L$ is given
cf.~\cite{FKO05} by
\begin{align}
&(\tilde LB)(\theta)\label{LtildeGL}\\
=&\int_{\R^d}dx\int_{\R^d}dy\,a(x-y)e^{-\phi(x-y)}(\theta(y)-\theta(x))
\delta B(\theta e^{-\phi(y-\cdot)}
+e^{-\phi(y-\cdot)}-1;x).\nonumber
\end{align}

\begin{theorem}\label{Th13} Given an $\alpha_0>0$, let
$B_0\in\mathcal{E}_{\alpha_0}$. For each $\alpha\in(0,\alpha_0)$ there is a
$T>0$ (which depends on $\alpha,\alpha_0$) such that there is a
unique solution $B_t$, $t\in[0,T)$, to the initial value problem
\eqref{evolgf}, \eqref{LtildeGL}, ${B_t}_{|t=0}= B_0$ in the space
$\mathcal{E}_\alpha$.
\end{theorem}

This theorem follows as a particular application of an abstract
Ovsjannikov-type result in a scale of Banach spaces which can be found
e.g.~in \cite[Theorem 2.5]{FKO11}, and the following estimate of norms.

\begin{proposition}\label{Proposition2}
Let $0<\alpha<\alpha_0$ be given. If $B\in\mathcal{E}_{\alpha''}$ for some
$\alpha''\in \left(\alpha,\alpha_0\right]$, then
$\tilde LB\in\mathcal{E}_{\alpha'}$ for all $\alpha\leq\alpha'<\alpha''$, and
we have
\[
\|\tilde LB\|_{\alpha'}\leq
2e^{\frac{\|\phi\|_{L^1}}{\alpha}}\|a\|_{L^1}\frac{\alpha_0}{\alpha''-\alpha'}\|B\|_{\alpha''}.
\]
\end{proposition}

To prove this result as well as other forthcoming ones the next lemma shows to
be useful.

\begin{lemma}\label{Lemma2} Let $\varphi,\psi:\R^d\times\R^d\to\R$ be such that,
for a.a.~$y\in\R^d$, $\varphi(y,\cdot)\in L^\infty:=L^\infty(\R^d)$,
$\psi(y,\cdot)\in L^1$ and $\|\varphi(y,\cdot)\|_{L^\infty}\leq c_0$,
$\|\psi(y,\cdot)\|_{L^1}\leq c_1$ for some constants $c_0,c_1>0$ independent
of $y$. For each $\alpha>0$ and all $B\in\mathcal{E}_\alpha$ let
\[
(L_0B)(\theta):=\int_{\R^d}dx\int_{\R^d}dy\,a(x-y)e^{-k\phi(x-y)}\left(\theta(y)-\theta(x)\right)\delta B(\varphi(y, \cdot)\theta+\psi(y,\cdot);x),
\]
$\theta\in L^1$. Here $a$ and $\phi$ are defined as before and $k\geq 0$ is a
constant. Then, for all $\alpha'>0$ such that $c_0\alpha'<\alpha$, we have
$L_0B\in\mathcal{E}_{\alpha'}$
and
\[
\|L_0B\|_{\alpha'}\leq 2e^{\frac{c_1}{\alpha}}\|a\|_{L^1}\frac{\alpha'}{\alpha-c_0\alpha'}\|B\|_{\alpha}.
\]
\end{lemma}

\begin{proof}
First we observe that from the considerations done in Subsection
\ref{Subsection2.2} it follows that $L_0B$ is an entire functional on $L^1$
and, in addition, that for all $r>0$, $\theta\in L^1$, and a.a.~$x,y\in\R^d$,
\begin{align*}
\left|\delta B(\varphi(y, \cdot)\theta+\psi(y,\cdot);x)\right|\leq
&\,
\left\|\delta B(\varphi(y, \cdot)\theta+\psi(y,\cdot);\cdot)\right\|_{L^\infty}\\
\leq &\, \frac{1}{r}\sup_{\|\theta_0\|_{L^1}\leq
r}\left|B(\varphi(y, \cdot)\theta+\psi(y,\cdot)+\theta_0)\right|,
\end{align*}
where, for all $\theta_0\in L^1$ such that $\|\theta_0\|_{L^1}\leq r$,
\[
\left|B(\varphi(y, \cdot)\theta+\psi(y,\cdot)+\theta_0)\right|\leq \|B\|_\alpha e^{\frac{\|\varphi(y, \cdot)\theta+\psi(y,\cdot)\|_{L^1}}{\alpha}+\frac{r}{\alpha}}
\leq \|B\|_\alpha e^{\frac{c_0\|\theta\|_{L^1}+c_1+r}{\alpha}}.
\]
As a result, due to the positiveness of $\phi$ and to the fact that $a$ is an
even function, for all $\theta\in L^1$ one has
\begin{align*}
|(L_0B)(\theta)|\leq &\,
\frac{1}{r}e^{\frac{c_0\|\theta\|_{L^1}+c_1+r}{\alpha}}
\|B\|_\alpha\int_{\R^d}dx\int_{\R^d}dy\,a(x-y)e^{-k\phi(x-y)}|\theta(y)-\theta(x)|\\
\leq &\,
\frac{2}{r}e^{\frac{c_1+r}{\alpha}}\|a\|_{L^1}\|\theta\|_{L^1}e^{\frac{c_0\|\theta\|_{L^1}}{\alpha}}\|B\|_\alpha.
\end{align*}
Thus,
\begin{align*}
\|L_0B\|_{\alpha'}&=\sup_{\theta\in L^1}\left(e^{-\frac{1}{\alpha'}\|\theta\|_{L^1}}|(L_0B)(\theta)|\right)\\
&\leq
\frac{2}{r}e^{\frac{c_1+r}{\alpha}}\|a\|_{L^1}\|B\|_\alpha\sup_{\theta\in
L^1}\left(e^{-\left(\frac{1}{\alpha'}-\frac{c_0}{\alpha}\right)\|\theta\|_{L^1}}\|\theta\|_{L^1}\right),
\end{align*}
where the supremum is finite provided
$\frac{1}{\alpha'}-\frac{c_0}{\alpha}>0$. In such a situation, the use of the
inequality $xe^{-mx}\leq\frac{1}{em}$, $x\geq0$, $m>0$ leads for each $r>0$ to
\begin{equation*}
\|L_0B\|_{\alpha'}\leq \frac{2}{r}\|a\|_{L^1}e^{\frac{c_1+r}{\alpha}}\frac{\alpha\alpha'}{e(\alpha-c_0\alpha')}\|B\|_\alpha.
\end{equation*}
The required estimate of norms follows by minimizing the expression
$\frac{1}{r}e^{\frac{c_1+r}{\alpha}}$ in the parameter $r$, that is, $r=\alpha$.
\end{proof}

\begin{proof}[Proof of Proposition \ref{Proposition2}]
In Lemma \ref{Lemma2} replace $\varphi$ by $e^{-\phi}$ and $\psi$ by
$e^{-\phi}-1$, and consider $k=1$. Due to the positiveness and integrability
properties of $\phi$ one has $e^{-\phi}\leq 1$ and
$|e^{-\phi}-1|=1-e^{-\phi}\leq \phi\in L^1$, ensuring the conditions
to apply Lemma \ref{Lemma2}.
\end{proof}

\begin{remark}
Concerning the initial conditions considered in Theorem \ref{Th13}, observe
that, in particular, $B_0$ can be an entire GF $B_{\mu_0}$ on $L^1$ such that,
for some constants $\alpha_0,C>0$,
$|B_{\mu_0}(\theta)|\leq C\exp(\frac{\|\theta\|_{L^1}}{\alpha_0})$ for all
$\theta\in L^1$. In such a situation an additional analysis is need in order
to guarantee that for each $t$ the local solution $B_t$ given by Theorem
\ref{Th13} is a GF (corresponding to some measure). For more details see
e.g.~\cite{FKO11,KoKuOl02} and references therein.
\end{remark}

\section{Vlasov scaling}\label{Subsection3.2}

We proceed to investigate the Vlasov-type scaling proposed in
\cite{FKK10a} for generic continuous particle systems and
accomplished in \cite{BKKK11} for the Kawasaki dynamics. As
explained in both references, we start with a rescaling of an
initial correlation function $k_0$, denoted by $k_0^{(\eps)}$, $\eps
>0$, which has a singularity with respect to $\eps$ of the type
$k_0^{(\eps)}(\eta) \sim \eps^{-|\eta|} r_0(\eta)$,
$\eta\in\Gamma_0$, being $r_0$ a function independent of $\eps$. The
aim is to construct a scaling of the operator $L$ defined in
\eqref{Vl41}, $L_\eps$, $\eps
>0$, in such a way that the following two conditions are fulfilled. The first
one is that under the scaling $L\mapsto L_\eps$ the solution
$k^{(\varepsilon)}_t$, $t\geq 0$, to
\begin{equation*}
\frac \partial {\partial t}k_t^{(\varepsilon)}=\hat L_\varepsilon^*k^{(\varepsilon)}_t,\quad {k^{(\varepsilon)}_t}_{|t=0}=k_0^{(\varepsilon)}
\end{equation*}
preserves the order of the singularity with respect to $\eps$, that
is, $k_t^{(\eps)}(\eta) \sim \eps^{-|\eta|} r_t(\eta)$,
$\eta\in\Gamma_0$. The second condition is that the dynamics $r_0
\mapsto r_t$ preserves the Lebesgue-Poisson exponents, that is, if
$r_0$ is of the form $r_0=e_\lambda(\rho_0)$, then each $r_t$,
$t>0$, is of the same type, i.e., $r_t=e_\lambda(\rho_t)$, where
$\rho_t$ is a solution to a non-linear equation (called a
Vlasov-type equation).

The previous scheme was accomplished in \cite{BKKK11} through the
scale transformation $\phi\mapsto\varepsilon\phi$ of the operator
$L$, that is,
\[
(L_\varepsilon F)(\gamma):=
\sum_{x\in\gamma}\int_{\R^d}dy\,a(x-y)e^{-\varepsilon E(y,\gamma)}\left(F(\gamma\setminus\{x\}\cup\{y\})-F(\gamma)\right).
\]
As shown in \cite[Example 12]{FKK10a}, \cite{BKKK11}, the
corresponding Vlasov-type  equation is given by
\begin{equation}
\frac{\partial}{\partial
t}\rho_t(x)=(\rho_t*a)(x)e^{-(\rho_t*\phi)(x)}
-\rho_t(x)(a*e^{-(\rho_t*\phi)})(x),\quad x\in\R^d,\label{Vl42}
\end{equation}
where $*$ denotes the usual convolution of functions. Existence of
classical solutions $0\leq\rho_t\in L^\infty$ to \eqref{Vl42} has
been discussed in \cite{BKKK11}. Therefore, it is natural to
consider the same scaling, but in GF.

To proceed towards GF, we consider $k^{(\varepsilon)}_t$ defined as before and
$k^{(\varepsilon)}_{t,\mathrm{ren}}(\eta):=\varepsilon^{|\eta|}k^{(\varepsilon)}_t(\eta)$. In terms of GF, these yield
\[
B_t^{(\eps)}(\theta):=\int_{\Gamma_0} d\lambda(\eta)
e_\lambda(\theta,\eta)k_t^{(\eps)}(\eta),
\]
and
\[
B_{t,\mathrm{ren}}^{(\varepsilon)}(\theta):=\int_{\Gamma_0}d\lambda(\eta)\,e_\lambda(\theta,\eta)k_{t,\mathrm{ren}}^{(\varepsilon)}(\eta)=\int_{\Gamma_0}d\lambda(\eta)\,e_\lambda(\varepsilon\theta,\eta)k_t^{(\varepsilon)}(\eta)=B_t^{(\varepsilon)}(\varepsilon\theta),
\]
leading, as in \eqref{obtevol}, to the initial value problem
\begin{equation}
\frac{\partial}{\partial t}B_{t,\mathrm{ren}}^{(\varepsilon)}
=\tilde L_{\varepsilon, \mathrm{ren}}B_{t,\mathrm{ren}}^{(\varepsilon)},\quad
{B^{(\varepsilon)}_{t,\mathrm{ren}}}_{|t=0}=B_{0,\mathrm{ren}}^{(\varepsilon)}.\label{V12}
\end{equation}

\begin{proposition}\label{taquase3}
For all $\varepsilon>0$ and all $\theta\in L^1$, we have
\begin{align}
(\tilde L_{\varepsilon, \mathrm{ren}}B)(\theta)&=
\int_{\R^d}dx\int_{\R^d}dy\,a(x-y)e^{-\varepsilon\phi(x-y)}(\theta(y)-\theta(x))\nonumber
\\ &\quad \times \delta B\left(\theta e^{-\varepsilon\phi(y-\cdot)} +\frac{e^{-\varepsilon\phi(y-\cdot)}-1}{\varepsilon};x\right).\label{taquase4}
\end{align}
\end{proposition}

\begin{proof} Since
\[
(\tilde L_{\varepsilon, \mathrm{ren}}B)(\theta)=
\int_{\Gamma_0}d\lambda(\eta)\,(\hat L_{\varepsilon, \mathrm{ren}}e_\lambda(\theta))(\eta)k(\eta),
\]
first we have to calculate
$(\hat L_{\varepsilon, \mathrm{ren}}e_\lambda(\theta))(\eta):=
\varepsilon^{-|\eta|}\hat L_\varepsilon(e_\lambda(\varepsilon\theta,\eta))$,
$\hat L_\varepsilon=K^{-1}L_\varepsilon K$ cf.~\cite{FKK10a}. Similar
calculations done in \cite[Subsection 4.2.1]{FKO05} show
\begin{align*}
(\hat L_{\varepsilon, \mathrm{ren}}e_\lambda(\theta))(\eta)
=& \,\sum_{x\in\eta}\int_{\R^d}dy\,a(x-y)e^{-\varepsilon\phi(x-y)}(\theta(y)-\theta(x))\\
&\quad\times e_\lambda\left(\theta
e^{-\varepsilon\phi(y-\cdot)}+\frac{e^{-\varepsilon\phi(y-\cdot)}-1}{\varepsilon},\eta\setminus\{x\}\right),
\end{align*}
and thus, using the relation between variational derivatives derived in
\cite[Proposition 11]{KoKuOl02}, one finds
\begin{align*}
(\tilde L_{\varepsilon, \mathrm{ren}}B)(\theta)=& \,\int_{\Gamma_0}d\lambda(\eta)\,k(\eta)\sum_{x\in\eta}\int_{\R^d}dy\,a(x-y)e^{-\varepsilon\phi(x-y)}(\theta(y)-\theta(x))\\
&\quad\times e_\lambda\left(\theta e^{-\varepsilon\phi(y-\cdot)}+\frac{e^{-\varepsilon\phi(y-\cdot)}-1}{\varepsilon},\eta\setminus\{x\}\right)\\
=& \,\int_{\R^d}dx\int_{\R^d}dy\,a(x-y)e^{-\varepsilon\phi(x-y)}(\theta(y)-\theta(x))\\&\quad\int_{\Gamma_0}d\lambda(\eta)\,k(\eta\cup\{x\})e_\lambda\left(\theta e^{-\varepsilon\phi(y-\cdot)}+\frac{e^{-\varepsilon\phi(y-\cdot)}-1}{\varepsilon},\eta\right)\\
=&
\,\int_{\R^d}dx\int_{\R^d}dy\,a(x-y)e^{-\varepsilon\phi(x-y)}(\theta(y)-\theta(x))\\&\quad\times
\delta B\left(\theta e^{-\varepsilon\phi(y-\cdot)}
+\frac{e^{-\varepsilon\phi(y-\cdot)}-1}{\varepsilon};x\right).\qedhere
\end{align*}
\end{proof}

\begin{proposition}\label{estimativas} (i) If $B\in\mathcal{E}_{\alpha}$ for
some $\alpha >0$, then, for all $\theta\in L^1$,
$(\tilde L_{\varepsilon,\mathrm{ren}}B)(\theta)$ converges as $\eps$ tends to
zero to
\[
(\tilde L_V B)(\theta):=
\int_{\R^d}dx\int_{\R^d}dy\,a(x-y)(\theta(y)-\theta(x))
\delta B(\theta -\phi(y-\cdot);x).
\]
(ii) Let $\alpha_0>\alpha>0$ be given. If $B\in\mathcal{E}_{\alpha''}$ for some
$\alpha''\in (\alpha, \alpha_0]$, then $\bigl\{\tilde L_{\varepsilon, \mathrm{ren}}B,\tilde{L}_V B\bigr\} \subset\mathcal{E}_{\alpha'}$ for all
$\alpha\leq\alpha'<\alpha''$, and we have
\[
\|\tilde{L}_{\#}B\|_{\alpha'}\leq
2\|a\|_{L^1}\frac{\alpha_0}{(\alpha''-\alpha')}
e^{\frac{\|\phi\|_{L^1}}{\alpha}}\|B\|_{\alpha''}
\]
where $\tilde{L}_{\#}=\tilde{L}_{\varepsilon, \mathrm{ren}}$ or
$\tilde{L}_{\#}=\tilde{L}_V$.
\end{proposition}

\begin{proof} (i) To prove this result we first analyze the pointwise
convergence of the variational derivative \eqref{taquase4} appearing in
$\tilde L_{\varepsilon, \mathrm{ren}}$. For this purpose we will
use the relation between variational derivatives derived in
\cite[Proposition 11]{KoKuOl02}, i.e.,
\[
\delta B(\theta_1+\theta_2;x)=\int_{\Gamma_0}d\lambda(\eta)\,
\delta^{|\eta|+1} B(\theta_1;\eta\cup\{x\})e_\lambda(\theta_2,\eta),\quad
a.a.\,x\in\R^d,\theta_1,\theta_2\in L^1,
\]
which allows to rewrite \eqref{taquase4} as
\begin{align}
&\delta B\left(\theta e^{-\varepsilon\phi (y -\cdot)}+
\frac{e^{-\varepsilon\phi (y -\cdot)}-1}{\varepsilon};x\right)\nonumber \\
=& \,\int_{\Gamma_0}\,d\lambda(\eta)\,
\delta^{|\eta|+1} B(\theta -\phi (y -\cdot);\eta\cup\{x\})\label{taquase7}\\
&\qquad\times
e_\lambda\left(\theta\left(e^{-\varepsilon\phi(y-\cdot)}-1\right) +
\frac{e^{-\varepsilon\phi(y-\cdot)}-1}{\varepsilon}+\phi(y-\cdot),\eta\right),\nonumber
\end{align}
for a.a.~$x,y\in\R^d$. Concerning the function
\[
f_\varepsilon:=f_\varepsilon(\theta,\phi,y):=\theta\left(e^{-\varepsilon\phi (y -\cdot )}-1\right)+\frac{e^{-\varepsilon\phi (y -\cdot )}-1}{\varepsilon}+\phi(y-\cdot)
\]
which appears in \eqref{taquase7}, for a.a.~$y\in\R^d$, one clearly
has $\lim_{\varepsilon\to 0}f_\varepsilon= 0$ a.e.~in $\R^d$. By
definition \eqref{taquase6}, the latter implies that
$e_\lambda(f_\varepsilon)$ converges $\lambda$-a.e.~to
$e_\lambda(0)$. Moreover, for the whole integrand function in
\eqref{taquase7}, estimates \eqref{2011!}, \eqref{Duarte} yield for
any $r>0$ and $\lambda$-a.a.~$\eta\in\Gamma_0$,
\begin{align*}
&\left|\delta^{|\eta|+1} B(\theta -\phi (y -\cdot);\eta\cup\{x\})e_\lambda(f_\varepsilon,\eta)\right|\\
\leq &\, \left\|\delta^{|\eta|+1} B(\theta -\phi (y -\cdot);\cdot)\right\|_{L^\infty(\R^{d(|\eta|+1)})} e_\lambda(|f_\varepsilon|,\eta)\\
\leq &\,
(|\eta|+1)!\left(\frac{e}{r}\right)^{|\eta|+1}e_\lambda(|f_\varepsilon|,\eta)\sup_{\|\theta_0\|_{L^1}\leq
r}
|B(\theta -\phi (y -\cdot)+\theta_0)|\\
\leq &\, (|\eta|+1)!\left(\frac{e}{r}\right)^{|\eta|+1}
e_\lambda(|\theta|+2|\phi(y-\cdot)|,\eta)e^{\frac{\|\theta-\phi(y-\cdot)\|_{L^1}+r}{\alpha}}\|B\|_\alpha
\end{align*}
with
\[
\int_{\Gamma_0}d\lambda(\eta)\,(|\eta|+1)!\left(\frac{e}{r}\right)^{|\eta|+1}
e_\lambda(|\theta|+2|\phi(y-\cdot)|,\eta)=
\sum_{n=0}^\infty(n+1)\left(\frac{e}{r}\right)^{n+1}\bigl(\|\theta\|_{L^1}+
2\|\phi\|_{L^1}\bigr)^n
\]
being finite for any $r> e (\|\theta\|_{L^1}+ 2\|\phi\|_{L^1})$.

As a result, by an application of the Lebesgue dominated convergence theorem
we have proved that, for a.a.~$x,y\in\R^d$, \eqref{taquase7} converges as
$\varepsilon$ tends to zero to
\[
\int_{\Gamma_0}\,d\lambda(\eta)\,
\delta^{|\eta|+1} B(\theta -\phi (y -\cdot);\eta\cup\{x\})
e_\lambda(0,\eta)= \delta B(\theta -\phi (y -\cdot);x).
\]
In addition, for the integrand function which appears in
$(\tilde L_{\varepsilon, \mathrm{ren}}B)(\theta)$ we have
\begin{align*}
&\left|a(x-y)e^{-\varepsilon\phi(x-y)}(\theta(y)-\theta(x))
\delta B\left(\theta e^{-\varepsilon\phi(y-\cdot)} +\frac{e^{-\varepsilon\phi(y-\cdot)}-1}{\varepsilon};x\right)\right|\\
\leq &\, \frac{e}{\alpha}a(x-y)|\theta(y)-\theta(x)|
\|B\|_\alpha\exp\left(\frac{1}{\alpha}\|\theta\|_{L^1}+\frac{1}{\alpha}\|\phi\|_{L^1}\right),
\end{align*}
for all $\varepsilon >0$ and a.a.~$x,y\in\R^d$, leading through a second
application of the Lebesgue dominated convergence theorem to the required
limit.

(ii) In Lemma \ref{Lemma2} replace $\varphi$ by
$e^{-\varepsilon\phi}$, $\psi$ by
$\frac{e^{-\varepsilon\phi}-1}{\varepsilon}$, and $k$ by
$\varepsilon$. Arguments similar to prove Proposition
\ref{Proposition2} complete the proof for $\tilde{L}_{\varepsilon,
\mathrm{ren}}$. A similar proof holds for $\tilde{L}_V$.
\end{proof}

Proposition \ref{estimativas} (ii) provides similar estimate of norms for
$\tilde{L}_{\varepsilon, \mathrm{ren}}$, $\eps>0$, and the limiting mapping
$\tilde{L}_V$. According to the Ovsjannikov-type result used to prove Theorem
\ref{Th13}, this means that given any
$B_{0,V},B_{0,\mathrm{ren}}^{(\eps)}\in\mathcal{E}_{\alpha_0}$, $\eps>0$, for
each $\alpha\in\left(0,\alpha_0\right)$ there is a $T>0$ such that there is a
unique solution
$B_{t,\mathrm{ren}}^{(\eps)}:\left[0,T\right)\to\mathcal{E}_\alpha$,
$\eps>0$, to each initial value problem \eqref{V12} and a unique solution
$B_{t,V}:\left[0,T\right)\to\mathcal{E}_\alpha$ to the initial value problem
\begin{equation}
\frac{\partial}{\partial t}B_{t,V}=\tilde L_{V}B_{t,V},\quad
{B_{t,V}}_{|t=0}=B_{0,V}.\label{V19}
\end{equation}
In other words, independent of the initial value problem under consideration,
the solutions obtained are defined on the same time-interval and with values in
the same Banach space. For more details see e.g.~Theorem 2.5 and its proof in
\cite{FKO11}. Therefore, it is natural to analyze under which conditions the
solutions to \eqref{V12} converge to the solution to \eqref{V19}. This
follows from a general result presented in \cite{FKO11} (Theorem 4.3).
However, to proceed to an application of this general result one needs the
following estimate of norms.

\begin{proposition}\label{Prop3}
Assume that $0\leq\phi\in L^1\cap L^\infty$ and let $\alpha_0>\alpha>0$ be
given. Then, for all $B\in\mathcal{E}_{\alpha''}$,
$\alpha''\in (\alpha, \alpha_0]$, the following estimate holds
 \begin{align*}
&\|\tilde L_{\varepsilon, \mathrm{ren}}B-\tilde L_V B\|_{\alpha'}\\
\leq &\,
2\eps\|a\|_{L^1}\|\phi\|_{L^\infty}\frac{e\alpha_0}{\alpha}\|B\|_{\alpha''}e^{\frac{\|\phi\|_{L^1}}{\alpha}}\biggl(
\left(2e\|\phi\|_{L^1}+\frac{\alpha_0}{e}\right)\frac{1}{\alpha''-\alpha'}+
\frac{8\alpha_0^2}{(\alpha''-\alpha')^2}\biggr)
\end{align*}
for all $\alpha'$ such that $\alpha\leq\alpha'<\alpha''$ and all $\eps>0$.
\end{proposition}

\begin{proof} First we observe that
\begin{align*}
&\left|(\tilde L_{\varepsilon, \mathrm{ren}}B)(\theta)-(\tilde L_V B)(\theta)\right|\leq \int_{\R^d}dx\int_{\R^d}dy\,a(x-y)\left|\theta(y)-\theta(x)\right|\\
&\times\left|e^{-\varepsilon\phi(x-y)}\delta B\left(\theta
e^{-\varepsilon\phi (y -\cdot )}+\frac{e^{-\varepsilon\phi (y -\cdot
)}-1}{\varepsilon};x\right)-\delta B\left(\theta -\phi (y -\cdot);x\right)\right|\\
\end{align*}
with
\begin{align}
&\left|e^{-\varepsilon\phi(x-y)}\delta B\left(\theta
e^{-\varepsilon\phi (y -\cdot )}+\frac{e^{-\varepsilon\phi (y -\cdot
)}-1}{\varepsilon};x\right)-\delta B\left(\theta -\phi (y -\cdot);x\right)\right|\nonumber\\
\leq &\, \left|\delta B\left(\theta e^{-\varepsilon\phi (y -\cdot
)}+\frac{e^{-\varepsilon\phi (y -\cdot
)}-1}{\varepsilon};x\right)-\delta B\left(\theta -\phi (y -\cdot);x\right)\right|\label{Ola}\\
&+\left(1-e^{-\varepsilon\phi(x-y)}\right)\left|\delta B\left(\theta
-\phi (y -\cdot);x\right)\right|.\nonumber
\end{align}
In order to estimate (\ref{Ola}), given any
$\theta_0,\theta_1,\theta_2\in L^1$, let us consider the function
$C_{\theta_0,\theta_1,\theta_2}(t)=dB\left(t\theta_1+(1-t)\theta_2;\theta_0\right)$, $t\in\left[0,1\right]$, where $dB$ is the first order differential of $B$,
defined in \eqref{taquase10}. One has
\begin{align*}
\frac{\partial}{\partial t} C_{\theta_0,\theta_1,\theta_2}(t) =&
\,\frac{\partial}{\partial
s} C_{\theta_0,\theta_1,\theta_2}(t+s)\Bigr|_{s=0}\\
=& \,\frac{\partial}{\partial
s} dB\bigl(\theta_2+t(\theta_1-\theta_2)+s(\theta_1-\theta_2);\theta_0\bigr)\Bigr|_{s=0}\\
=& \,\frac{\partial^2}{\partial s_1\partial s_2}B\bigl(\theta_2+t(\theta_1-\theta_2)+s_1(\theta_1-\theta_2)+s_2\theta_0\bigr)\Bigr|_{s_1=s_2=0}\\
=&
\,\int_{\mathbb{R}^d}dx\int_{\mathbb{R}^d}dy\,(\theta_1(x)-\theta_2(x))\theta_0(y)\,\delta^2
B(\theta_2+t(\theta_1-\theta_2);x,y),
\end{align*}
leading to
\begin{align*}
&\bigl|dB(\theta_1;\theta_0)-dB(\theta_2;\theta_0)\bigr|\\
=& \,\bigl|C_{\theta_0,\theta_1,\theta_2}(1)-C_{\theta_0,\theta_1,\theta_2}(0)\bigr|\\
\leq &\,
\max_{t\in[0,1]}\int_{\mathbb{R}^d}dx\int_{\mathbb{R}^d}dy\,
\left|\theta_1(x)-\theta_2(x)\right|\left|\theta_0(y)\right|
\left|\delta^2 B(\theta_2+t(\theta_1-\theta_2);x,y)\right|\\
\leq &\,
\|\theta_1-\theta_2\|_{L^1}\|\theta_0\|_{L^1}\max_{t\in[0,1]}\|\delta^2
B(\theta_2+t(\theta_1-\theta_2);\cdot)\|_{L^\infty(\R^{2d})},
\end{align*}
where, through estimate \eqref{Duarte} with $r=\alpha''$,
\[
\|\delta^2 B(\theta_2+t(\theta_1-\theta_2);\cdot)\|_{L^\infty(\R^{2d})}
\leq 2\frac{e^3}{\alpha''^2}\|B\|_{\alpha''}\exp\left(\frac{\|\theta_2+t(\theta_1-\theta_2)\|_{L^1}}{\alpha''}\right).
\]
As a result,
\begin{align*}
&\bigl|dB(\theta_1;\theta_0)-dB(\theta_2;\theta_0)\bigr|\\
\leq &\,
2\frac{e^3}{\alpha''^2}\|\theta_1-\theta_2\|_{L^1}\|\theta_0\|_{L^1}\|B\|_{\alpha''}\max_{t\in[0,1]}\exp\left(\frac{t\|\theta_1\|_{L^1}+(1-t)\|\theta_2\|_{L^1}}{\alpha''}\right),
\end{align*}
for all $\theta_0,\theta_1,\theta_2\in L^1$. In particular, this shows that
for all $\theta_0\in L^1$,
\begin{align*}
&\left|dB\left(\theta e^{-\varepsilon\phi (y -\cdot
)}+\frac{e^{-\varepsilon\phi (y -\cdot
)}-1}{\varepsilon};\theta_0\right)-dB\left(\theta -\phi (y -\cdot);\theta_0\right)\right|\\
\leq &\,
2\varepsilon\frac{e^3}{\alpha''^2}\|\phi\|_{L^\infty}\|B\|_{\alpha''}
\left(\|\theta\|_{L^1}+ \|\phi\|_{L^1}\right)\|\theta_0\|_{L^1}\\
&\qquad\times\max_{t\in[0,1]}\exp\left(\frac{1}{\alpha''}\left(t\left(\|\theta\|_{L^1}+
\|\phi\|_{L^1}\right)+(1-t)\left(\|\theta\|_{L^1}+\|\phi\|_{L^1}\right)\right)\right)\\
=&
\,2\varepsilon\frac{e^3}{\alpha''^2}\|\phi\|_{L^\infty}\|B\|_{\alpha''}
\left(\|\theta\|_{L^1}+ \|\phi\|_{L^1}\right)
\exp\left(\frac{1}{\alpha''}\left(\|\theta\|_{L^1}+\|\phi\|_{L^1}\right)\right)\|\theta_0\|_{L^1},
\end{align*}
where we have used the inequalities
\begin{align*}
\|\theta e^{-\varepsilon\phi (y -\cdot )}-\theta\|_{L^1}&\leq\eps\|\phi\|_{L^\infty}\|\theta\|_{L^1},\\
\Bigl\|\frac{e^{-\varepsilon\phi (y -\cdot
)}-1}{\varepsilon}+\phi (y -\cdot
)\Bigr\|_{L^1}&\leq\eps\|\phi\|_{L^\infty}\|\phi\|_{L^1},\\
\Bigl\|\theta e^{-\varepsilon\phi (y -\cdot )}+\frac{e^{-\varepsilon\phi (y -\cdot
)}-1}{\varepsilon}\Bigr\|_{L^1}&\leq\|\theta\|_{L^1}+\|\phi\|_{L^1}.
\end{align*}
In other words, we have shown that the norm of the bounded linear functional
on $L^1$
\[
L^1\ni\theta_0\mapsto dB\left(\theta e^{-\varepsilon\phi (y -\cdot )}+\frac{e^{-\varepsilon\phi (y -\cdot
)}-1}{\varepsilon};\theta_0\right)-dB\left(\theta -\phi (y -\cdot);\theta_0\right)
\]
is bounded by
\[
Q:=2\varepsilon\frac{e^3}{\alpha''^2}\|\phi\|_{L^\infty}\|B\|_{\alpha''}
\left(\|\theta\|_{L^1}+ \|\phi\|_{L^1}\right)
\exp\left(\frac{1}{\alpha''}\left(\|\theta\|_{L^1}+\|\phi\|_{L^1}\right)\right).
\]
Since this operator norm is given by
\[
\left\|\delta B\left(\theta e^{-\varepsilon\phi (y -\cdot )}+\frac{e^{-\varepsilon\phi (y -\cdot
)}-1}{\varepsilon};\cdot\right)-\delta B\left(\theta -\phi (y -\cdot);\cdot\right)\right\|_{L^\infty}
\]
cf.~Subsection \ref{Subsection2.2}, this means that
\[
\left\|\delta B\left(\theta e^{-\varepsilon\phi (y -\cdot )}+\frac{e^{-\varepsilon\phi (y -\cdot
)}-1}{\varepsilon};\cdot\right)-\delta B\left(\theta -\phi (y -\cdot);\cdot\right)\right\|_{L^\infty}\leq Q.
\]

In this way we obtain
\begin{align*}
&\left|(\tilde L_{\varepsilon, \mathrm{ren}}B)(\theta)-(\tilde L_V B)(\theta)\right|\\
\leq &\,  \int_{\R^d}dx\int_{\R^d}dy\,a(x-y)\left|\theta(y)-\theta(x)\right|\\
&\times\left\{\left\|\delta B\left(\theta e^{-\varepsilon\phi (y
-\cdot )}+\frac{e^{-\varepsilon\phi (y -\cdot
)}-1}{\varepsilon};\cdot\right)-\delta B\left(\theta -\phi (y -\cdot);\cdot\right)\right\|_{L^\infty}\right.\\
&\qquad\left.+\varepsilon\|\phi\|_{L^\infty}\left\|\delta
B\left(\theta -\phi (y
-\cdot);\cdot\right)\right\|_{L^\infty}\vphantom{\frac{e^{-\varepsilon\phi
(y -\cdot
)}-1}{\varepsilon}}\right\}\\
\leq &\,
2\varepsilon\|\phi\|_{L^\infty}\|a\|_{L^1}\frac{e}{\alpha''}
\exp\left(\frac{1}{\alpha''}\left(\|\theta\|_{L^1}+\|\phi\|_{L^1}\right)\right)
\|\theta\|_{L^1}\\
&\qquad\qquad\times
\left\{2\frac{e^2}{\alpha''}\left(\|\theta\|_{L^1}+
\|\phi\|_{L^1}\right) +1\right\}\|B\|_{\alpha''},
\end{align*}
and thus
\begin{align*}
&\|\tilde L_{\varepsilon, \mathrm{ren}}B-\tilde L_V B\|_{\alpha'}\\
\leq &\,
2\varepsilon\|\phi\|_{L^\infty}\|a\|_{L^1}\frac{e}{\alpha''}e^{\frac{\|\phi\|_{L^1}}{\alpha''}}\left\{2\frac{e^2}{\alpha''}\sup_{\theta\in
L^1}\left(\|\theta\|_{L^1}^2
\exp\left(\|\theta\|_{L^1}\left(\frac{1}{\alpha''}-\frac{1}{\alpha'}\right)\right)\right)\right.\\
&\left.+\left(2\frac{e^2}{\alpha''}\|\phi\|_{L^1}+1\right)\sup_{\theta\in
L^1}\left(\|\theta\|_{L^1}\exp\left(\|\theta\|_{L^1}\left(\frac{1}{\alpha''}-\frac{1}{\alpha'}\right)\right)\right)\right\}\|B\|_{\alpha''},
\end{align*}
and the proof follows using the inequalities
$xe^{-mx}\leq\frac{1}{me}$ and $x^2e^{-mx}\leq\frac{4}{m^2e^2}$ for $x\geq0$,
$m>0$.
\end{proof}

We are now in conditions to state the following result.

\begin{theorem}\label{Props1} Given an $0<\alpha<\alpha_0$, let
$B_{t,\mathrm{ren}}^{(\varepsilon)}, B_{t,V}$,
$t\in\left[0,T\right)$, be the local solutions in $\mathcal{E}_{\alpha}$ to
the initial value problems \eqref{V12}, \eqref{V19} with
$B_{0,\mathrm{ren}}^{(\varepsilon)},B_{0,V}\in\mathcal{E}_{\alpha_0}$. If
$0\leq\phi\in L^1\cap L^\infty$ and $\lim_{\varepsilon\rightarrow 0}\|B_{0,\mathrm{ren}}^{(\varepsilon)}-B_{0,V}\|_{\alpha_0}=0$, then, for each
$t\in\left[0,T\right)$,
\[
\lim_{\varepsilon\rightarrow 0}\|B_{t,\mathrm{ren}}^{(\varepsilon)}-B_{t,V}\|_{\alpha}=0.
\]
Moreover, if
$B_{0,V}(\theta)=\exp\left(\int_{\mathbb{R}^d}dx\,\rho_0(x)\theta(x)\right)$,
$\theta\in L^1$, for some function $0\leq\rho_0\in L^\infty$ such that
$\|\rho_0\|_{L^\infty}\leq \frac{1}{\alpha_0}$, then for each
$t\in\left[0,T\right)$,
\begin{equation}\label{expsol}
B_{t,V}(\theta)=\exp\left(\int_{\mathbb{R}^d}dx\,
\rho_t(x)\theta(x)\right), \quad \theta\in L^1,
\end{equation}
where $0\leq\rho_t\in L^\infty$ is a classical solution to the equation
\eqref{Vl42}.
\end{theorem}

\begin{proof}
The first part follows directly from Proposition~\ref{Prop3} and
\cite[Theorem 4.3]{FKO11}, taking in \cite[Theorem 4.3]{FKO11} $p=2$ and
\[
N_\eps=2\eps\|a\|_{L^1}\|\phi\|_{L^\infty}\frac{e\alpha_0}{\alpha}e^{\frac{\|\phi\|_{L^1}}{\alpha}}\max\left\{2e\|\phi\|_{L^1}+\frac{\alpha_0}{e},8\alpha_0^2\right\}.
\]

Concerning the last part, we begin by observing that it has been shown in
\cite[Subsection 4.2]{BKKK11} that given a $0\leq \rho_0\in L^\infty$ such that
$\|\rho_0\|_{L^\infty}\leq\frac{1}{\alpha_0}$, there is a solution
$0\leq\rho_t\in L^\infty$ to \eqref{Vl42} such that
$\|\rho_t\|_{L^\infty} \leq \frac{1}{\alpha_0}$. This implies that
$B_{t,V}$, given by \eqref{expsol}, does not leave the initial Banach space
$\mathcal{E}_{\alpha_0}\subset\mathcal{E}_\alpha$. Then, by an argument of
uniqueness, to prove the last assertion amounts to show that $B_{t,V}$ solves
equation \eqref{V19}. For this purpose we note that for any
$\theta,\theta_1\in L^1$ we have
\[
\frac{\partial}{\partial z_1}
B_{t,V}(\theta+z_1\theta_1)\biggr|_{z_1=0}=B_{t,V}(\theta)\int_{\mathbb{R}^d}dx
\rho_t(x)\theta_1(x),
\]
and thus $\vd B_{t,V}(\theta;x)= B_{t,V}(\theta) \rho_t(x)$. Hence, for all
$\theta\in L^1$,
\begin{align*}
(\tilde L_V B_{t,V})(\theta)=& \,
B_{t,V}(\theta)\left(\int_{\R^d}dx\int_{\R^d}dy\,a(x-y)
\left(\theta(y)-\theta(x)\right)\rho_t(x)e^{-(\rho_t*\phi)(y)}\right)\\
=&
\,B_{t,V}(\theta)\left(\int_{\R^d}dy\,\theta(y)\left(a*\rho_t\right)(y)
e^{-(\rho_t*\phi)(y)}\right.\\
&\qquad\qquad\left.-\int_{\R^d}dx\,\theta(x)\left(a*e^{-(\rho_t*\phi)(y)}\right)(x)\rho_t(x)\right).
\end{align*}
Since $\rho_t$ is a classical solution to \eqref{Vl42}, $\rho_t$ solves a weak
form of equation \eqref{Vl42}, that is, the right-hand side of the latter
equality is equal to
\[
B_{t,V}(\theta)\frac{d}{dt}\int_{\mathbb{R}^d}dx\,
\rho_t(x)\theta(x) = \frac{\partial}{\partial t}
B_{t,V}(\theta).\qedhere
\]
\end{proof}

\end{document}